\documentclass[1p]{elsarticle}

\usepackage{lineno,hyperref}
\modulolinenumbers[5]
\usepackage{fix-cm}
\usepackage{mathtools}
\usepackage{amsmath,amsthm}
\usepackage{graphicx}
\usepackage{latexsym,enumerate,amssymb,amsbsy}
\usepackage{graphics,color}
\usepackage{verbatim}
\usepackage{url}
\usepackage{dblfloatfix}
\usepackage{mathptmx}
\usepackage{tikz-network}
\usepackage{booktabs}
\usepackage{xfrac}
\usepackage{amsfonts}
\usepackage{stmaryrd}
\usepackage{algorithm}
\usepackage[noend]{algpseudocode}
\usepackage{url}
\usepackage{colortbl}
\usepackage{adjustbox}
\usepackage{array,longtable}
\usepackage{caption}
\usepackage{subcaption}
\usepackage{arydshln}

\newtheorem{proposition}{Proposition}
\newtheorem{theorem}[proposition]{Theorem}
\newtheorem{corollary}[proposition]{Corollary}
\newtheorem{lemma}[proposition]{Lemma}

\newtheorem{example}{Example}

\newtheorem{remark}{Remark}

\newtheorem{definition}{Definition}

\def\F{{\mathbb{F}}}

\def\Z{{\mathbb{Z}}}

\def\0{{\mathbf{0}}}
\def\1{{\mathbf{1}}}

\def\mdc{$\Gamma({\bf N}, S)$  }

\DeclarePairedDelimiter{\floor}{\lfloor}{\rfloor}

\newcommand\dsb[1]{\llbracket #1 \rrbracket}

\newcommand{\blue}[1]{{\color{blue}#1}}

\allowdisplaybreaks

\begin{document}
	
	\begin{frontmatter}
		\title{New Qubit Codes from Multidimensional Circulant Graphs \tnoteref{t1}}
		\tnotetext[t1]{This material is based upon work supported by the National Science Foundation under Grant DMS-2243991.}

		\author[tamuc]{Padmapani Seneviratne\corref{cor1}}
		\ead{Padmapani.Seneviratne@tamuc.edu}	
		
		\author[TC]{Hannah Cuff\fnref{fn1}}
		\ead{hannah.cuff@trincoll.edu}
		
		\author[CU]{Alexandra Koletsos\fnref{fn1}}
		\ead{ak4749@columbia.edu}

		\author[SC]{Kerry Seekamp\fnref{fn1}}
		\ead{kseekamp@smith.edu}

		\author[PU]{Adrian Thananopavarn\fnref{fn1}}
		\ead{adrianpt@princeton.edu}

		\address[tamuc]{Department of Mathematics, Texas A\&M University-Commerce,\\
			2600 South Neal Street, Commerce, TX 75428.}
		
		\address[TC]{Department of Mathematics, Trinity College, 300 Summit Street Hartford, CT 06106.} 
		
		\address[CU]{Department of Mathematics, Columbia University, 2990 Broadway New York, NY 10027.}
		
		\address[SC]{Department of Mathematical Sciences, Smith College, 10 Elm Street Northampton, MA 01063.}
		
		\address[PU]{Department of Mathematics, Princeton University, Fine Hall, Washington Road Princeton, NJ 08544-1000. }


		\begin{abstract}
			Two new qubit stabilizer codes with parameters  $\dsb{77, 0, 19}_2$ and  $\dsb{90, 0, 22}_2$ are constructed for the first time by employing additive symplectic self-dual $\F_4$ codes from multidimensional circulant (MDC) graphs. We completely classify MDC graph codes for lengths $4\le n \le 40$ and show that many optimal $\dsb{\ell, 0, d}$ qubit codes can be obtained from the MDC construction. 
			Moreover, we prove that adjacency matrices of MDC graphs have nested block circulant structure and determine isomorphism properties of MDC graphs.
			
		\end{abstract}
		
	\end{frontmatter}
	
	
	\section{Introduction}\label{Intro}
	Errors in quantum computing present a unique challenge in storing and transmitting data. When designing quantum error-correcting codes (QECCs), one needs to address both bit flip and phase flip errors which can occur simultaneosly. The existence of a QECC, which can  protect quantum information against decoherence, was first introduced in 1995 by Shor~\cite{shor}.
	In their seminal work~\cite{Calderbank1998}, Calderbank, Rains, Shor, and Sloane established a connection between classical error-correcting codes  and binary QECCs (qubit codes).
	
	Unlike their traditional counterparts, zero-dimensional qubit codes play an important role in quantum computing. They can be used to test the accuracy of quantum computers.  They can also be used to test the storage locations of  qubits that are experiencing greater effects of decoherence than originally predicted~\cite{Calderbank1998}. Zero-dimensional qubit codes corresponds to self-dual additive codes over the finite field $\F_4 = \{0, 1, \omega, \omega^2\}$, where $\omega^2 = \omega + 1$. Danielsen and Parker~\cite{Danielsen2006} employed Schlingemanns~\cite{schl} work to demonstrate that every graph generates a symplectic self-dual additive code and conversely, any self-dual additive code over $\F_4$ has a graphical representation.
	
	Highly symmetrical and vertex transitive graphs, such as circulant graphs, have been extensively studied for their ability to generate optimal self-dual additive codes~\cite{Grassl2017,verbanov, Saito2019}. 
	However, other classes of graphs have also demonstrated success in generating self-dual additive codes. For example, a recent paper~\cite{seneviratne}, presents five new $\dsb{\ell, 0, d}$ qubit codes with parameters $(\ell, d) \in \{(78, 20), (90, 21), (91, 22), (93, 21), (96, 22)$ produced by metacirculant graphs, a class of vertex transitive graphs.

	We further extend Danielsen and Parkers work by studying multidimensional circulant graphs (MDCs), a generalization of circulant graphs on multiple coordinates. We define MDC graphs and study their properties in Section~\ref{sec:mdc}. Enumeration of self-dual codes of length $4 \le n \le 40$ constructed from MDC graphs are outlined in Section~\ref{sec:additive}. Finally, we present new and optimal qubit codes from multi-dimensional construction in Section~\ref{sec:new}.
	
	\section{Multidimensional Circulant Graphs}\label{sec:mdc}
	
	Circulant graphs have been extensively studied for their applications in the field of code theory. Historically, they have demonstrated success in generating QECCs via self-dual additive codes. We reference the definition of circulant graphs presented in~\cite{Monakhova2012}.

	\begin{definition}
		Let $\mathbb{Z}_{n}$ denote the ring of integers modulo $n$. A circulant graph $C(n, S)$ is a Cayley graph on $\Z_n$. 
		That is, it is a graph whose vertices are labeled $\{0, 1, \cdots, n-1 \}$, where two vertices $x$ and $y$ are adjacent if and only if $x - y \,( mod\ n) \in S$, where $S \subset \Z_n$ with $S = -S$ and $0 \notin S$.
	\end{definition}
	
	The adjacency matrix of a circulant graph is a {\it circulant matrix}. An $n \times n$ matrix $B$ is circulant if it has the form
	\begin{equation}~\label{eq:one}
		B = 
		\begin{pmatrix}
			b_1 & b_2 & \cdots & b_{n-1} & b_n \\
			b_n & b_1 & \cdots & b_{n-2} & b_{n-1} \\
			\vdots  & \vdots  & \ddots & \vdots & \vdots \\
			b_3 & b_4 & \cdots & b_1 &b_2 \\
			b_2 & b_3 & \cdots & b_n &b_1 
		\end{pmatrix}.
	\end{equation}	
	
	There are numerous  generalizations of circulant graphs in literature~\cite{Alspach1982, Monakhova2012}.  Leighton~\cite{Leighton} extended the notion of circulant graphs to multiple coordinates, defining these graphs as multidimensional circulant. Essentially, a  MDC graph  is a Cayley graph on $\Z_{n_1} \times Z_{n_2}\times \cdots \times \Z_{n_k}$.

	\begin{definition} 
		Let $ {\bf N} = (n_1, n_2, \ldots, n_k)$ and let $S  \subset \Z_{n_1} \times Z_{n_2}\times \cdots \times \Z_{n_k}$, with $ S= -S$ and ${\bf 0} \notin S$.
		A MDC graph \mdc has the vertex set $V(\Gamma) =\{(v_1,  \ldots, v_k): v_1 \in \Z_{n_1},  \ldots, v_k \in \Z_{n_k}\}$ and two vertices ${\bf x} = (x_1,, \cdots, x_k)$ and ${\bf y} = (y_1, , \cdots, y_k)$ are adjacent if and only if $(x_1 - y_1\ (mod\ n_1)), \cdots, x_k - y_k\ (mod\ n_k)) \in S$,
		
	\end{definition}
	
	The following example presents the difference in structure between MDC and circulant graphs.
	
	\begin{example}
		The hypercube graph $Q_3$ is not circulant, but it is MDC, with parameters\\ 
		$\Gamma((2, 4), \{(0, 1), (0, 3), (1, 0)\})$. 
		The vertex set is partitioned into 
		\[
		V_0 :=\{(0,0),(0,1),(0,2),(0,3)\} \mbox{ and }
		V_1 :=\{(1,0),(1,1),(1,2),(1,3)\}.
		\]

		\begin{figure}[H] \label{fig:Q3}
			\centering
			\begin{tikzpicture}[multilayer=3d]
				\SetLayerDistance{-1.5}
				\Plane[x=-.7,y=-1.5,width=4.5,height=3,color=gray,layer=2,NoBorder]
				\Plane[x=-.7,y=-1.5,width=4.5,height=3,NoBorder]
				
				\Vertex[x=-0.1,y=-0.1, IdAsLabel,layer=1]{1}
				\Vertex[x=0.7,y=1,IdAsLabel,layer=1]{2}
				\Vertex[x=2.5,y=0.2,IdAsLabel,layer=1]{3}
				\Vertex[x=1.6,y=-1,IdAsLabel,layer=1]{4}
				\Vertex[IdAsLabel,color=gray,layer=2]{5}
				\Vertex[x=0.8,y=1,IdAsLabel,color=gray,layer=2]{6}
				\Vertex[x=2.4,y=0.4,IdAsLabel,color=gray,layer=2]{7}
				\Vertex[x=1.75,y=-1,IdAsLabel,color=gray,layer=2]{8}
				
				\Edge[color=blue](1)(2)
				\Edge[color=blue](2)(3)
				\Edge[color=blue](3)(4)
				\Edge[color=blue](4)(1)
				
				\Edge[style=dashed](5)(6)
				\Edge[style=dashed](6)(7)
				\Edge(7)(8)
				\Edge(8)(5)
				
				\Edge(1)(5)
				\Edge[style=dashed](2)(6)
				\Edge(4)(8)
				\Edge(3)(7)
				\Edge[style=dashed](6)(7)
			\end{tikzpicture}
			\caption{The $3$-cube graph as a MDC graph}
		\end{figure}
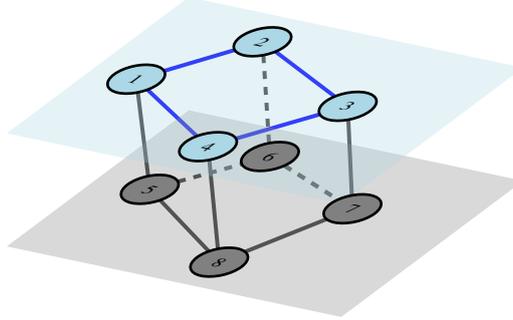
		Figure $1$ represents the hypercube graph $Q_3$ as a MDC graph with parameters $N=(2, 4)$ and $S = \{(0, 1), (0, 3), (1, 0)\})$. The upper layer contains the vertices in $V_0$ as $1, \ldots, 4$ and the lower layer presents the vertices in $V_1$ as $5, \ldots, 8$.
	\end{example}
	
	\begin{definition}
		A nested block-circulant matrix is a $n \times n$ matrix that takes the form 
		\begin{equation*}
			A = 
			\begin{pmatrix}
				B_1 & B_2 & \ldots & B_{{l_0}-1} & B_{l_0} \\ 
				B_{l_0} & B_1 & \ldots & B_{{l_0}-2} & B_{{l_0}-1} \\ 
				\vdots & \vdots & \ddots & \vdots & \vdots \\ 
				B_{3} & B_4 & \ldots & B_{1} & B_{2} \\ 
				B_{2} & B_{3} & \ldots & B_{{l_0}} & B_{1}
			\end{pmatrix} 
		\end{equation*}   
		where each block $B_1,B_2,\ldots,B_{l_0}$ can be recursively partitioned into blocks $B^{i}_1,B^{i}_2,\ldots,B^{i}_{l_1}$, where $1\le i \le r$ for some $r$ and the smallest form of each block is a circulant matrix, as given in equation~\ref{eq:one}.
	\end{definition}

	\begin{theorem}\label{thm:nested}
		Consider a multidimensional circulant graph $\Gamma (\mathbf{N},S)$ with $\mathbf{N}=(n_1,n_2,\ldots,n_k)$, where $n_1 \le n_2 \le \ldots \le n_k$. First, let us define $N_1 = \frac{(n_2 \cdot n_3 \cdots n_k)}{n_1}$, $N_2 = \frac{N_1}{n_2}$, $N_3 = \frac{N_2}{n_3}$, $\ldots$ , $N_k = \frac{N_{k-1}}{n_k}$. 
		The adjacency matrix $A(\Gamma)$ of \mdc is nested block circulant and has the form:
		\begin{equation*}
			A(\Gamma) = 
			\begin{pmatrix}
				A_{1,1} & A_{1,2} & \ldots & A_{1,{l-1}} & A_{1,l} \\ 
				A_{1,l} & A_{1,1} & \ldots & A_{1,{l-2}} & A_{1,{l-1}} \\ 
				\vdots& \vdots & \ddots & \vdots & \vdots \\ 
				A_{1,3} & A_{1,4} & \ldots & A_{1,1} & A_{1,2} \\ 
				A_{1,2} & A_{1,3} & \ldots & A_{1,l} & A_{1,1}
			\end{pmatrix} 
		\end{equation*}   
		where each block $A_{1,j}$ is a $N_1 \times N_1$
		submatrix of $A(\Gamma)$ for $1 \le j \le n_1 $. 
	\end{theorem}
	
	\begin{proof}Let $\Gamma(\mathbf{N},S)$ be a multidimensional circulant graph where $\mathbf{N}=(n_1,n_2,\cdots,n_k)$ and the vertices are ordered lexicographically. Let $N_1 = \frac{(n_2 \cdot n_3 \cdots n_k)}{n_1}$, $N_2 = \frac{N_1}{n_2}$, $N_3 = \frac{N_2}{n_3}$, $\ldots$ , $N_k = \frac{N_{k-1}}{n_k}$. 
		
		First, let us show that the adjacency matrix $A$ of $\Gamma(\mathbf{N},S)$ is block-circulant. To do so, let us partition the vertex set as $V = V^1_0 \cup V^1_1 \cup \cdots \cup V^1_{n_1-1}$, where \begin{equation} \label{eqn:vertex}
			V^1_{i} = \{(i,x_2,\ldots,x_k)\ |\ 0 \le x_j \le n_j - 1,\ 2 \le j \le k\}\; \mbox{and}\; 0 \le i \le n_{1}-1.
		\end{equation}
		
		Accordingly, each $V^1_i$ will form $n_1$ submatrices, each of which we will call $A^{1}_{i,l}$, where $|A^{1}_{i,l}| = N_1$, $0 \le l \le n_{1}-1$. 
		Now, let vertex ${\bf x} = (i, x_2, \ldots, x_k) \in V_{i}^{1}$ and vertex ${\bf y} = (l, y_2, \ldots, y_{k}) \in V_{l}^{1}$. We know ${\bf x} \sim {\bf y}$ iff ${\bf x} - {\bf y} \in S$, or rather, $(i-l,x_2-y_2,\ldots,x_k-y_k) \in S$. With this in mind, consider vertex ${\bf x}^{\prime} = (i+1, x_2, \ldots, x_k) \in V_{i+1}^{1}$ and vertex ${\bf y}^{\prime} = (l+1, y_2, \ldots, y_{k}) \in V_{l+1}^{1}$. We can see that ${\bf x}^{\prime} \sim {\bf y}^{\prime}$ iff $(i-l,x_2-y_2,\ldots,x_k-y_k) \in S$. Therefore, the adjacency relation between $V^1_i$ and $V^1_l$ is the same as the adjacency relation between $V^1_{i+1}$ and $V^1_{l+1}$. Thus, we can establish the equality $A^1_{i,l} = A^1_{{i+1},{l+1}}$, showing that each row of submatrices in $A(\Gamma)$ is a cyclic shift of one block to the right of the previous row of submatrices.
		
		Next, let us show that each submatrix $A^1_{i,l}$ is also block-circulant. Let us partition $V^1_i = V^2_{i,0}\cup V^2_{i,1} \cup \ldots \cup V^2_{i,n_2-1}$, where 
		\begin{center}
			$V^2_{i,p} = \{(i,p,x_3,\ldots,x_k)\ |\ 0 \le x_j \le n_j - 1,\ 3 \le j \le k\}$ and $0 \le p \le n_2-1$. 
		\end{center}
		Accordingly, each $V^2_{i,p}$ will form $n_2$ submatrices, each of which we will call $A^{2}_{p,q}$, where $|A^{2}_{p,q}| = N_2$, $0 \le q \le n_{2}-1$.
		Now, let vertex ${\bf x} = (i, p, x_3, \ldots, x_k) \in V^2_{i,p}$ and vertex ${\bf y} = (l, q, y_3, \ldots, y_{k}) \in V^2_{l,q}$. We know ${\bf x} \sim {\bf y}$ iff ${\bf x} - {\bf y} \in S$, or rather, $(i-l,p-q,x_3-y_3,\ldots,x_k-y_k) \in S$. With this in mind, consider vertex ${\bf x}^{\prime} = (i, p+1, x_3, \ldots, x_k) \in V^2_{i,p+1}$ and vertex ${\bf y}^{\prime} = (l, q+1, y_3, \ldots, y_{k}) \in V^2_{l,q+1}$. We can see that ${\bf x}^{\prime} \sim {\bf y}^{\prime}$ iff $(i-l,p-q,x_3-y_3,\ldots,x_k-y_k) \in S$. Thus, we can establish the equality $A^{2}_{p,q} = A^{2}_{p+1,q+1}$, showing that each row of submatrices in $A^1_{i,l}$ is a cyclic shift of one block to the right of the previous row of submatrices in $A^1_{i,l}$.
		
		Thus, using this argument, there will be $|N|$ layers $l$ of nested blocks that compose $A$, and layers $1 \le l \le N-1$ will be block-circulant with ${n_l}^2$  submatrices $A^l_{i,j}$, where $0 \le i,j \le n_l-1$ and $|A^l_{i,j}|=N_l$. The $|N|^{th}$ layer will be circulant since $|A^{l}_{i,j}|=N_k=1$ where $l=|N|$, showing that it will be composed of blocks with dimension $1 \times 1$ circulating around the matrix, which are really just singular elements.
	\end{proof}
	
	\begin{example}\label{eg:nested}
		Let $\Gamma (\mathbf{N},S)$ be a MDC graph with $\mathbf{N} = (3,2,2)$ and $S=\{(0,1,0),(2,0,1),(2,0,0)\}$. This yields the adjacency matrix 
		\begin{center} 
			$A$ = $\begin{pmatrix} \begin{array}
					{|c c:c c|} \hline 0 & 0 & 1 & 0 \\ 0 & 0 & 0 & 1 \\ \hdashline 1 & 0 & 0 & 0 \\ 0 & 1 & 0 & 0 \\\hline
				\end{array} & \begin{array}
					{c c c c} 1 & 1 & 0 & 0  \\ 1 & 1 & 0 & 0  \\ 0 & 0 & 1 & 1 \\ 0 & 0 & 1 & 1 
				\end{array} & \begin{array}
					{c c c c} 1 & 1 & 0 & 0  \\ 1 & 1 & 0 & 0  \\ 0 & 0 & 1 & 1 \\ 0 & 0 & 1 & 1 
				\end{array} \\\\ \begin{array}
					{c c c c} 1 & 1 & 0 & 0  \\ 1 & 1 & 0 & 0  \\ 0 & 0 & 1 & 1 \\ 0 & 0 & 1 & 1 
				\end{array} & \begin{array}
					{c c c c} 0 & 0 & 1 & 0 \\ 0 & 0 & 0 & 1 \\ 1 & 0 & 0 & 0 \\ 0 & 1 & 0 & 0 
				\end{array} & \begin{array}
					{c c c c} 1 & 1 & 0 & 0  \\ 1 & 1 & 0 & 0  \\ 0 & 0 & 1 & 1 \\ 0 & 0 & 1 & 1 
				\end{array} \\\\ \begin{array}
					{c c c c} 1 & 1 & 0 & 0  \\ 1 & 1 & 0 & 0  \\ 0 & 0 & 1 & 1 \\ 0 & 0 & 1 & 1 
				\end{array} & \begin{array}
					{c c c c} 1 & 1 & 0 & 0  \\ 1 & 1 & 0 & 0  \\ 0 & 0 & 1 & 1 \\ 0 & 0 & 1 & 1 
				\end{array} & \begin{array}
					{c c c c} 0 & 0 & 1 & 0 \\ 0 & 0 & 0 & 1 \\ 1 & 0 & 0 & 0 \\ 0 & 1 & 0 & 0 
				\end{array}
			\end{pmatrix}$ \end{center}
		Note that $A(\Gamma)$ is composed of $9$ blocks that are $4 \times 4$ sub-matrices, where the first row of blocks circulates to the right to produce the second and third rows of blocks. Each of these larger blocks is composed of four $2 \times 2$ smaller blocks, in accordance with Theorem~\ref{thm:nested}.
	\end{example}
	
	\begin{remark}
		The adjacency matrix of a MDC graph can be completely determined by the first row of each of the $n_1$ first row block matrices. In example~\ref{eg:nested}, supports of the first row of each of the first three block matrices are $\{3\}, \{1, 2\},$ and  $\{1, 2\}$. We could write $\Gamma((3, 2, 2), [\{3\}, \{1, 2\}, \{1, 2\}])$ rather than listing all elements of the defining set $S$. Henceforth, $S$ denotes the supports of the first row of the first $n_1$ blocks.
	\end{remark}

	The complement $G^{\prime}$ of a graph $G$ is the graph on the same vertex set as $G$ such that two distinct vertices in $G^{\prime}$ are adjacent if and only if they are not adjacent in $G$. It is a well known fact that the complement graph of a circulant graph is again a circulant graph. We will show that the same is true for MDC graphs as well. 
	
	\begin{lemma}\label{lem:complement}
		The complement graph of a MDC graph is also MDC.
	\end{lemma}
	
	\begin{proof}
		Let \mdc be a MDC graph with the vertex set $V$ and let $V^{*} = V \setminus {\bf 0}$. Clearly, $S \subset V^{*}$. Let us define $S^{\prime} = V^{*}\setminus S$ and claim $\Gamma^{\prime}({\bf N}, S) = \Gamma({\bf N}, S^{\prime})$. Let ${\bf x}$ and ${\bf y}$ be two adjacent vertices in \mdc. By the definition, ${\bf x} - {\bf y} \in S$ and  ${\bf x} - {\bf y} \notin V^{*}\setminus S = S^{\prime}$. Hence, if ${\bf x}$ is adjacent to ${\bf y}$ in \mdc then they are not adjacent in 	$\Gamma({\bf N}, S^{\prime}) $ and vice versa. Therefore, we have 
		$\Gamma^{\prime}({\bf N}, S) = \Gamma({\bf N}, S^{\prime})$
		
	\end{proof}

	An $m$-multipartite graphs is a graph in which vertices can be partitioned into $m$ distinct independent sets.  
	\begin{theorem}
		Let \mdc be a MDC graph and define
		$S_{j} = \bigcup\limits_{(x_1, \ldots, x_k)\in S} \{x_j\}$, the set consisting of all $j^{th}$ coordinates of elements of $S$.  
		Then \mdc  is a multipartite graph whenever $0 \notin S_1$.
	\end{theorem}
	
	\begin{proof}
		Arrange the vertices in lexicographical order as in equation~\ref{eqn:vertex}. 
		\begin{equation*}
			V_{i} = \{(i,x_2,\ldots,x_k)\ |\ 0 \le x_j \le n_j - 1,\ 2 \le j \le k\}\; \mbox{and}\; 0 \le i \le n_{1}-1.
		\end{equation*}
		Then $V(\Gamma) = V_{0}\cup V_{1} \cup \cdots \cup V_{n_{j}-1} $ is a partition of the vertex set. Assume $0 \notin S_1$  and
		let ${\bf v}=(i, v_2, \ldots, v_{k})$ and ${\bf u} = (i, u_2, \ldots, u_{k})$ be arbitrary vertices in the same partition $V_{i}$. 
		Note that ${\bf v} - {\bf u} =  (0, v_2-u_2, \ldots, v_k-u_k)$ and ${\bf v} - {\bf u} \notin S$ as $0 \notin S_1$. Therefore, none of the vertices in same partition $V_i$ are connected to each other and \mdc is a $n_1$ partite graph.
		
	\end{proof}
	
	Next, we establish a relationship between MDC graphs and another class of vertex transitive graphs, namely metacirculant graphs, introduced by Alspach and Parsons \cite{Alspach1982}. First, we recall the definition of metacirculant graphs.
	
	\begin{definition}~\cite{Alspach1982}
		Let $m,n$ be two fixed positive integers and $\alpha\in \mathbb{Z}_{n}$ be a unit. Let $S_{0}, S_{1}, \ldots, S_{\floor{m/2}}  \\ \subset\mathbb{Z}_{n}$ satisfy the four properties $S_{0} = -S_{0}$,  $0\notin S_{0}$,  $\alpha^{m}S_k = S_k$ for $1\le k\le \lfloor m/2 \rfloor$, and If $m$ is even then $\alpha^{m/2}S_{m/2} = -S_{m/2}$. 
		The meta-circulant graph $\Gamma:=\Gamma\left(m, n, \alpha, S_{0}, S_{1}, \ldots, S_{\floor{m/2}}\right)$ has the vertex set $V(\Gamma)=\mathbb{Z}_{m}\times \mathbb{Z}_{n}$. Let $V_0, V_1, \ldots, V_{m-1}$, where 
		$V_{i} := \{(i, j) : 0\le j \le n-1 \}$, be a partition of $V(\Gamma)$. Let $1 \le k \le \floor{m/2}$. Vertices $(i, j)$ and $(i+k, h)$ are adjacent if and only if $(h-j) \in \alpha^{i} \, S_k$.
	\end{definition}
	
	The following result shows that any two dimensional MDC graph can be represented as a metacirculant graph. But, the converse of the theorem is not true in general.

	\begin{theorem}\label{thm:meta}
		Any two-dimensional MDC graph can be represented as a metacirculant graph with $\alpha = 1$.
	\end{theorem}
	
	\begin{proof}
		We consider two-dimensional MDC graphs $\Gamma((m,n), S)$ and metacirculant graphs\\ $\Gamma(m,n,1,S_0,S_1,\cdots S_k)$.  By definition, both vertex sets are $\mathbb{Z}_m \times \mathbb{Z}_n$, so they have a trivial bijective mapping. We wish to show that this bijective mapping results in a homomorphism between edges.
		
		Given any two-dimensional circulant graph $G_1 = \Gamma((m,n), S)$, let $S_i = \{s | (i,s) \in S\}$ for $0 \leq i \leq \lfloor m/2 \rfloor$, and $G_2$ be the corresponding metacirculant graph with $\alpha = 1$. We first verify that $G_2$ satisfies the criteria for metacirculant graphs:
		$$S = -S \implies S_0 = -S_0$$
		$$0 \notin S \implies 0 \notin S_0$$
		$$\alpha^mS_k = 1^mS_k = S_k$$
		$$S = -S \implies S_{m/2} = -S_{m/2} \text{ for even } m$$
		
		Now, for any two vertices $(a_1, b_1), (a_2,b_2) \in \mathbb{Z}_m \times \mathbb{Z}_n$, we have
		\begin{align*}
			(a_1, b_1) \sim (a_2,b_2) \text{ in } G_1 &\iff (a_2 - a_1,b_2-b_1) \in S \\
			&\iff b_2-b_1 \in S_{a_2 - a_1} \text{ or } b_1-b_2 \in S_{a_1 - a_2} \\
			&\iff b_2-b_1 \in \alpha^{a_1}S_{a_2 - a_1} \text{ or } b_1-b_2 \in \alpha^{a_2}S_{a_1 - a_2} \\
			&\iff (a_1, b_1) \sim (a_2,b_2) \text{ in } G_2,
		\end{align*}
		thus demonstrating that adjacency is the same and $G_1$ and $G_2$ are isomorphic.
	\end{proof}
	
	Now that we have established fundamental properties of MDC graphs, we explore some isomorphism properties of these graphs. Understanding isomorphism properties within the class of MDC graphs and between the classes of MDC graphs and other vertex transitive graphs allows us to improve the search process for self-dual additive codes from MDC graphs, effectively reducing the number of MDC graphs in the search process for a given length. We now present these isomorphism properties.

	Certain MDC graphs happen to be, in fact, one dimensional circulant graphs. Leighton~\cite{Leighton} proved the following characterization of MDC graphs. 
	
	\begin{proposition}~\cite{Leighton}\label{thm:prime}
		The MDC graph \mdc is isomorphic to a circulant graph on $n$ vertices whenever $n = \prod_{1}^{k}n_i$ with	${\bf N} = (n_1, n_2, \ldots, n_k)$ and $n_i, 1\le i \le k$ are distinct primes.
	\end{proposition}
	
	Using theorem~\ref{thm:prime} as a framework, we establish a more stronger result valid for all distinct relatively prime numbers.

	\begin{theorem}\label{thm:coprime}
		The MDC graph \mdc is isomorphic to a circulant graph on $n$ vertices whenever $n = \prod_{1}^{k}n_i$ with	${\bf N} = (n_1, n_2, \ldots, n_k)$ and $n_i, 1\le i \le k$ are distinct relatively prime integers.
	\end{theorem}
	
	\begin{proof}
		We will prove the result for $k = 2$. Let $m$ and $n$ be relatively prime integers. Define the function $\phi$ such that
		\begin{eqnarray*}
			\phi : \Z_m \times \Z_n \rightarrow \Z_{mn}\\
			\phi (x, y) = nx + my		
		\end{eqnarray*}
		Clearly, $\phi$ is a bijection from $\Z_m \times \Z_n $ to $\Z_{mn}$. Let $\Gamma = \Gamma((m, n), S)$ be a MDC graph indexed by the vertex set 
		$\Z_m \times \Z_n $. We will show that there exists a corresponding isomorphic circulant graph $C(mn, \phi(S))$.\\ 
		Given $S$ is a defining set for the MDC graph $\Gamma((m, n), S)$, we first need to show that $\phi(S)$ is a defining set for the circulant graph  $C(mn, \phi(S))$. 
		Let $(s_1, s_2), (t_1, t_2)\in S$ such that $(s_1, s_2) +  (t_1, t_2) = (0, 0)$. Then $\phi((s_1, s_2)) +  \phi((t_1, t_2)) = (ns_1+ms_2)+(nt_1+mt_2) = n(s_1+t_1) + m(s_2+t_2) = 0$. 
		Consequently, if $(t_1, t_2)$ is the additive inverse of $(s_1, s_2)\in S$ then $ \phi((t_1, t_2))$ is the additive inverse of  $\phi((s_1, s_2)) \in \phi(S)$, and $\phi(S)$ is a defining set for $C(mn, \phi(S))$.\\
		Let ${\bf x} =(x_1, x_2)$ and ${\bf y} = (y_1, y_2)$ be two adjacent vertices in $\Gamma$. Then ${\bf x} - {\bf y} = (x_1 - y_1, x_2 - y_2) = (s_1, s_2) \in S$.  Now $\phi(s_1, s_2) = ns_1+ms_2 \in \phi(S)$. 
		Further, $\phi(x_1, x_2) = nx_1 + mx_2$, $\phi(y_1, y_2) = ny_1 + my_2$, and $\phi(x_1, x_2) - \phi(y_1, y_2) = (nx_1 + mx_2) - (ny_1 + my_2) = n(x_1-y_1) + m(x_2-y_2) = ns_1+ms_2 \in \phi(S)$. Therefore, we have shown that if two vertices are adjacent in the MDC graph $\Gamma$ then their images under $\phi$ are adjacent in the circulant graph $C(mn, \phi(S))$ implying $\Gamma((m,n, S) \equiv C(mn, \phi(S))$. 
		, 
	\end{proof}
	
		This allows us to establish a pattern that determines the total number of distinct sets $\mathbf{N}$ for a multidimensional circulant graph with a particular number of vertices.
		\begin{corollary}
			The number of non-isomorphic sets $\mathbf{N}$ for a MDC graph $\Gamma (N,S)$ with a particular number of vertices $n$ is equivalent to the number of ways to represent $n$ as a product of prime powers. 
		\end{corollary}
		
		\begin{example}
			Consider  $n=36$: The prime factorization of $n = 2^2\cdot 3^2$. Then there are $4$ distinct sets that generate non-isomorphic MDC graphs.
			\begin{center}
				$n = 3^2 * 2^2  \rightarrow N = (9,4) \leftrightarrow N = (36)$\\
				\vspace{.1cm}
				$n = 3^1*3^1*2^1*2^1 \rightarrow N = (3,3,2,2) \leftrightarrow (6,3,2) \leftrightarrow (6,6)$\\
				\vspace{.1cm}
				$n = 3^2*2^1*2^1 \rightarrow N = (9,2,2) \leftrightarrow (18,2)$\\
				\vspace{.1cm}
				$n = 3^1*3^1*2^2 \rightarrow (3,3,4) \leftrightarrow  (3,12)$\\
			\end{center}
			
		\end{example}

	We also can observe isomorphisms among graphs with the same vertex set. One such family of isomorphic graphs is given by the following theorem.
	\begin{theorem}
		
		The multidimensional circulant graphs \mdc and $\Gamma({\bf N}, \sigma(S))$ are isomorphic, where 
		\begin{align*}
			\sigma\colon \Z_{n_1}\times \Z_{n_2}\times \cdots \Z_{n_k} & \longrightarrow\mathbf \Z_{n_1}\times \Z_{n_2}\times \cdots \Z_{n_k} \\[-1ex]
			(a_1, a_2, \ldots, a_k) & \longmapsto (\alpha_{1}a_{1}, \alpha_{2}a_{2}, \ldots, \alpha_{k}a_{k}).
		\end{align*}
		and $\alpha_1\in \Z_{n_1}, \alpha_2 \in \Z_{n_2}, \ldots \alpha_{k}\in \Z_{n_k}$ are units.
	\end{theorem}

	\begin{proof}
		Because all $\alpha$'s are units, $\sigma$ has a well-defined inverse,
		$\sigma^{-1}((a_1,  \ldots, a_k)) = (\alpha_{1}^{-1}a_{1}, \ldots, \alpha_{k}^{-1}a_{k})$ and is bijective. 
		Let ${\bf a} = (a_1, \ldots, a_k)$ and ${\bf b} = (b_1, \ldots, b_k)$ be two vertices of \mdc. We know  ${\bf a} \sim {\bf b}$ if and only if ${\bf a} - {\bf b} \in S$ and ${\bf a} - {\bf b}  = (a_1-b_1, \ldots, a_k - b_k)\in S$. Observe that $\sigma({\bf a}) - \sigma({\bf b}) =  (\alpha_1a_1-\alpha_1b_1, \ldots, \alpha_ka_k-\alpha_kb_k)  = (\alpha_1(a_1-b_1), \ldots, \alpha_k(a_k-b_k)) \in \sigma(S)$, implying $\sigma({\bf a}) \sim \sigma({\bf b})$ in $\Gamma({\bf N}, \sigma(S))$.
	\end{proof}

		In some special cases, there may be an isomorphism between different  $\mathbf{N}$'s even though the elements of $\mathbf{N}$ are not coprime. One example is the cube graph:
		\begin{example}
			The $3$-cube graph can be represented in two different ways: $\Gamma((4,2),\{(0,1),(1,0),(3,0)\})$ and $\Gamma((2,2,2),\{(0,0,1),(0,1,0),(1,0,0)\})$.
		\end{example}
		
		This isomorphism may be generalized to more dimensions, according to the following theorem:
		\begin{theorem}
			Every $k$-dimensional circulant graph with $\mathbf{N} = (4,2,2,\dotsc,2)$ is isomorphic to a $(k+1)$-dimensional circulant graph with $\mathbf{N} = (2,2,2,2,\dotsc,2)$.
		\end{theorem}
		\begin{proof}
			Let $\Gamma ((4,2,2,\dotsc,2),S)$ be a multidimensional circulant graph. Now, we define a transformation $\phi(s): \mathbb{Z}_4 \times \mathbb{Z}_2^{k-1} \rightarrow \mathbb{Z}_2^{k+1}$ as follows: If $s = (s_1, s_2,\cdots,s_k)$, then let
			\[ \phi(s) = \begin{cases} 
				(0,0,s_2,\cdots,s_k) & s_1 = 0 \\
				(1,0,s_2,\cdots,s_k) & s_1 = 1 \\
				(1,1,s_2,\cdots,s_k) & s_1 = 2\\
				(0,1,s_2,\cdots,s_k) & s_1 = 3.
			\end{cases}
			\]
			As $\phi$ forms a bijective mapping on the first two coordinates and leaves the rest the same, $\phi$ is bijective on the vertex sets from $\mathbf{N} = (4,2,2,\dotsc,2)$ to $\mathbf{N} = (2,2,2,2,\dotsc,2)$.
			
			Now, using this mapping, we will show that applying $\phi$ to the elements of $S$ results in a homomorphism. Suppose two vertices $v_1$ and $v_2$ are adjacent in $\Gamma((4,2,2,\dotsc,2),S)$. Then, in $\Gamma ((2,2,2,2,\dotsc,2),\phi(S))$,  we have 3 cases.
			\begin{enumerate}
				\item If the first coordinates of $v_1$ and $v_2$ are the same, the first two coordinates of $\phi(v_1)$ and $\phi(v_2)$ are the same, so $(0,0,s_2,\cdots,s_k) \in S$ connects the two in $\Gamma ((2,2,2,2,\dotsc,2),\phi(S))$.
				\item If the first coordinates of $v_1$ and $v_2$ differ by one, exactly one of the first two coordinates of $\phi(v_1)$ and $\phi(v_2)$ will differ, so either $(1,0,s_2,\cdots,s_k) \in S$ or $(0,1,s_2,\cdots,s_k) \in S$ connects the two in $\Gamma ((2,2,2,2,\dotsc,2),\phi(S))$.
				\item If the first coordinates of $v_1$ and $v_2$ differ by two, both of the first two coordinates of $\phi(v_1)$ and $\phi(v_2)$ will differ, so $(1,1,s_2,\cdots,s_k) \in S$ connects the two in $\Gamma ((2,2,2,2,\dotsc,2),\phi(S))$.
			\end{enumerate}
			In any case, $\phi(v_1)$ and $\phi(v_2)$ are adjacent, and as the cases exactly cover the coordinates of $\phi(v_1)$ and $\phi(v_2)$, if $\phi(v_1)$ and $\phi(v_2)$ are adjacent, then $v_1$ and $v_2$ will be as well. Therefore, $\Gamma ((4,2,2,\dotsc,2),S)$ is isomorphic to $((2,2,2,2,\dotsc,2),\phi(S))$. 
		\end{proof}

	
	\section{Self-dual additive codes from MDC graphs}\label{sec:additive}
	
		In this section, we use our own unique construction of multidimensional circulant graphs to generate zero-dimensional quantum error correcting codes, which we represent as self-dual additive codes over $\F_4$.
		Most of the qubit codes of length $n=1$ through $n=30$ found on \url{http://www.codetables.de} are extremal, meaning these codes meet their appropriate bound. However, we attempt to improve the minimum distance of non-extremal codes using MDC construction.
		We find two new $0$-dimensional qubit codes of lengths $77$ and $90$ with respective minimum distances $19$ and $22$, improving upon previous best-known minimum distances by $1$.

	An additive code $C$ of length $n$ over $\F_4 $ is an additive subgroup of $\F_{4}^{n}$.
	An element ${\bf c}$ of $C$ is called a codeword of $C$. The weight of a vector ${\bf u} \in \F_{4}^{n}$ is the number of nonzero entries of ${\bf u}$. The least nonzero weight of all codewords in $C$ is called the minimum distance of $C$. If $C$ is an additive code of length $n$ over $\F_4$ with minimum distance $d$ and size $2^k$, then $C$ is denoted by  $(n, 2^k, d)_4$. 
	The weight distribution of $C$ is the set $\{W_0, W_1, \ldots, W_r \}$, where $0\le r \le n$ and $W_{j} = W_{j}(C)$ is the number of codewords of weight $j$ in $C$.

	Given two vectors ${\bf u} = (u_1, u_2, \ldots, u_n)$ and ${\bf v} = (v_1, v_2, \ldots, v_n)$ in $\F_{4}^{n}$, the Hermitian trace inner product of ${\bf u}$ and ${\bf v}$ is defined by 
	\begin{eqnarray*}
		{\bf u}*{\bf v} = \sum_{i = 1}^{n} u_{i}v_{i}^{2} + u_{i}^{2}v_{i}.
	\end{eqnarray*}
	
	The symplectic dual $C^{*}$ of an additive code $C$ is given by $C^{*} = \{{\bf u}\in \F_{4}^{n} \mid {\bf u}*{\bf c} = 0\; \mbox{for all}\; {\bf c}\in C\}$. An additive code $C$ is called symplectic self-dual if $C = C^{*}$.

	Schlingemann~\cite{schl} and later Danielsen~\cite{Danielsen2006} showed that every self-dual additive code over $\F_4$ can be represented by a graph. In particular, if 
	$C(\Gamma)$ denotes the additive code generated by the row span of the matrix $A(\Gamma) + \omega\cdot I$, where $A(\Gamma)$ is the adjacency matrix of a graph $\Gamma$ and $I$ is the identity matrix, then $C(\Gamma)$ is symplectic self-dual.
	
	There were several studies involving classification of additive self-dual codes over $\F_4$. First, Danielsen and Parker~\cite{Danielsen2006} did a complete classification  for lengths $n\le 12$. Later, Gulliver and Kim~\cite{GulKim}, Grassl and Harada~\cite{Grassl2017}, and Saito~\cite{Saito2019} contributed to the classification of lengths up to $n \le 50$.  These studies centered on additive codes from one-dimensional circulant graphs with circulant adjacency and bordered matrices. In this work, we consider multidimensional circulant graphs and expand the search space for new qubit codes.
	
	An additive self-dual code $C$ over $\F_4$ is called Type $II$ if the weights of all the codewords in $C$ is even. A code which is not Type $II$ is called Type $I$. Any Type $II$ code must have even length. 
	We classify Type $I$ and Type $II$ additive self-dual codes from MDC graphs by the following result.
	
	\begin{lemma}\label{lem:size}
		Let $\Gamma = $ \mdc be a MDC graph that produces the additive self-dual code $C_{\Gamma}$. Then $C_{\Gamma}$ is Type $II$ if and only if $\mid S \mid$ is odd.
	\end{lemma}
	
	\begin{proof}
		It was shown in~\cite{Danielsen2006} that self-dual additive code $C$  generated by a regular graph $\Gamma$ is Type $II$ if and only if all vertices of $\Gamma$ have odd degrees.  MDC graphs are regular with valency $|S|$ and the result follows.
	\end{proof}

	We ran an exhaustive search to generate every possible self-dual additive code that could be generated with MDC graphs when $n$ ranges from $4$ to $40$. To conclude this section, we present a table that provide a comparative study between the qubit codes generated with MDC graphs, and qubit codes generated with circulant graphs.
	
	In the table $1$, $n$ denotes the number of vertices (length of the additive code), $N$ signifies distinct sets, $d_{\textit{max}}^{mdc}(n)$ indicates the maximum minimum distance among self-dual additive codes generated by MDC graphs, $d_{\textit{max}}^{c}(n)$ denotes the maximum known minimum distance generated using circulant graphs \cite{Saito2019}, and $d_{max}(n,0)$ means the maximum known minimum distance among all qubit codes with distance $n$ and dimension $0$~\cite{Grassl:codetables}.
	
	\begin{table}[hbt!]
		\centering
		\setlength{\tabcolsep}{2pt}
		\begin{tabular}{|ccccc|ccccc|}
			\hline
			$n$ & $N$ & $d_{max}^{mdc}(n)$ & $d_{\textit{max}}^{c}(n)$ & $d_{max}(n,0)$ & $n$ & $N$ & $d_{\textit{max}}^{mdc}(n)$ & $d_{\textit{max}}^{c}(n)$ & $d_{max}(n,0)$ \\
			\hline
			{4} & [2, 2] & 2 & 2 & 2 & {24} & [3, 8],[2, 12],[2,2, 6]& 8 & 8 & 8-10\\[1.5pt]				 			
			{6} & [2, 3] & 4 & 4 & 4 & {26} & [2, 13] & 8 & 8 & 8-10\\[1.5pt] 
			{8} & [2, 4],[2, 2, 2] & 4 & 4 & 4 & 	{27} & [3, 9] & 8 & 8 & {9-10} \\
			{} & {}& {}& {}&{}&{} & 					[3, 3, 3] & 6 & {} & {}\\ [1.5pt] 										 			
			{9} & [3, 3] & 4, 3& 4 & 4  & 		{28} & [4, 7],[2, 14] & 10, 8  & 10 & 10 \\ [1.5pt] 								 				
			{10} & [2, 5] & 4 & 4 & 4 & {30} & [3, 10] & 12 & 12 & 12\\[1.5pt] 						
			{12} & [3, 4],[2, 6] & 6,4 & 6 & 6 & {32}& [2, 16] & 10 & 10 & 10-12 \\
			{ } & { }& { } & { } & { } & {} & [8,4],[2, 2, 8][2, 2, 2, 4] & 8 &  { } & { }\\
			{ } & { }& { } & { } & { } & { } & [2, 4, 4] & 6 &{ } &{ }\\
			{ } & { }& { } & { } & { } & { } & [2, 2, 2, 2, 2] & 8 & {} & {} \\ [1.5pt] 
			{14} & [2, 7] & 6 & 6 & 6 &		{33} & [3, 11] & 10 & 10 & 10-12\\ [1.5pt] 
			{15} & [3, 5] & 6 & 6 & 6 & 	{34} & [2, 17] & 10  & 10 & 10-12\\ [1.5pt] 
			{16} & [2, 8] & 6 & 6 & 6 & 	{35} & [5, 7] & 10 & 10 &  11-13\\ 	
			{} & [4, 4],[2, 2, 4] & 4 & {} & {} & {} & {} & {} & {} &{}\\[1.5pt] 
			{ } & [2, 2, 2, 2] & 4 & { } & { } & {} & {} & {} & {} & {}\\[1.5pt]							
			{18} & [2, 9], [3, 6] & 6 & 6 & 8 & {36} & [2, 18] & 12 & 11 & 12-14\\[1.5pt]
			{} & {} & {} & {} & {} & {} & 			[4, 9]  & 11 & {} & {}\\ 														
			{ } & {}&{} & { } & { }& {} & [3, 12],[6, 6] & 10 & {} & {}\\[1.5pt]
			{20} & [4, 5], [2, 10] & 8 & 8 & 8 & 	{38} & [2, 19] & 12 & 12 & 12-14\\[1.5pt] 												
			{21} & [3, 7] & 7 & 7 & 8 & {39} & [3, 13] & 11 & 11 & 11-14\\[1.5pt] 
			{22} & [2, 11] & 8 & 8 & 8 & {40} &	[5, 8],[2, 20],[2,2, 10] & 12 & 12 & 12-14\\ 											
			\hline								
		\end{tabular} 
		\caption{Minimum distances of MDC graph codes.}
		
	\end{table}

	\newpage
	
	\section{New and optimal qubit codes}\label{sec:new}
	This section presents new and optimal symplectic self-dual additive codes generated from the MDC construction. We refer to a code $C$ as new if the minimum distance of $C$ is higher than the best known minimum distance available in the literature for the same length and the dimension. A code $C$ is called optimal if $C$ has the best known minimum distance among known codes with the same parameters.

		A central point of focus of our work involves comparing the quantum codes generated by multidimensional circulant graphs with the quantum codes generated by circulant graphs. We have thus far shown that MDC graphs have the potential to produce codes with better minimum distances than codes produced by circulant graphs. For example, the best minimum distance of codes generated by circulant graphs for $n=36$ is $11$, while the best minimum distance of codes generated by MDC graphs for $n=36$ is $12$. In the following proposition, we establish two non-isomorphic families of MDC graphs that yield $\dsb{36,0,12}$ qubit codes, adapting Proposition 2 in ~\cite{seneviratne}.

	\begin{proposition}
		The non-isomorphic MDC graphs
		\begin{equation*}
			\begin{split}
				\Gamma_{36, 1} & =\Gamma((2,18), [\{ 4, 5, 6, 7, 13, 14, 15, 16 \}, \{ 1, 3, 7, 13, 17 \}]\\
				\Gamma_{36, 2} & =\Gamma((2,18), [\{ 3, 7, 13, 17 \}, \{ 1, 5, 6, 8, 9, 11, 12, 14, 15 \}]
			\end{split}
		\end{equation*}	
		generates two inequivalent 	 $(36,2^{36},12)$, Type II additive self-dual codes $C_{36, 1}$ and $C_{36, 2}$ yielding two inequivalent $\dsb{36, 0, 12}$ qubit codes. 
		
	\end{proposition}
	
	\begin{proof}
		The MDC graphs $\Gamma_{36, 1}$ and $\Gamma_{36, 2}$ have valency $|S| = 13$ and by Lemma~\ref{lem:size}, \blue{the} self-dual additive codes $C_{36, 1}$ and $C_{36, 2}$ are Type $II$. 
		The weights $W_{12}(C_{36,1}) = 28764$ and $W_{12}(C_{36,2}) = 20844$ implying the two codes $C_{36, 1}$ and $C_{36, 2}$ are inequivalent. 
		By Theorem~\ref{thm:meta}, the MDC graphs $\Gamma_{36, 1}$ and $\Gamma_{36, 2}$ are isomorphic to metacirculant graphs $G_{36,1}$ and $G_{36, 2}$ respectively, obtained in~\cite{seneviratne}, Proposition 2. 
	\end{proof}
	
	\begin{figure}[ht]
		\centering
		\subcaptionbox*{i.}[.45\linewidth]{%
			\includegraphics[width=\linewidth]{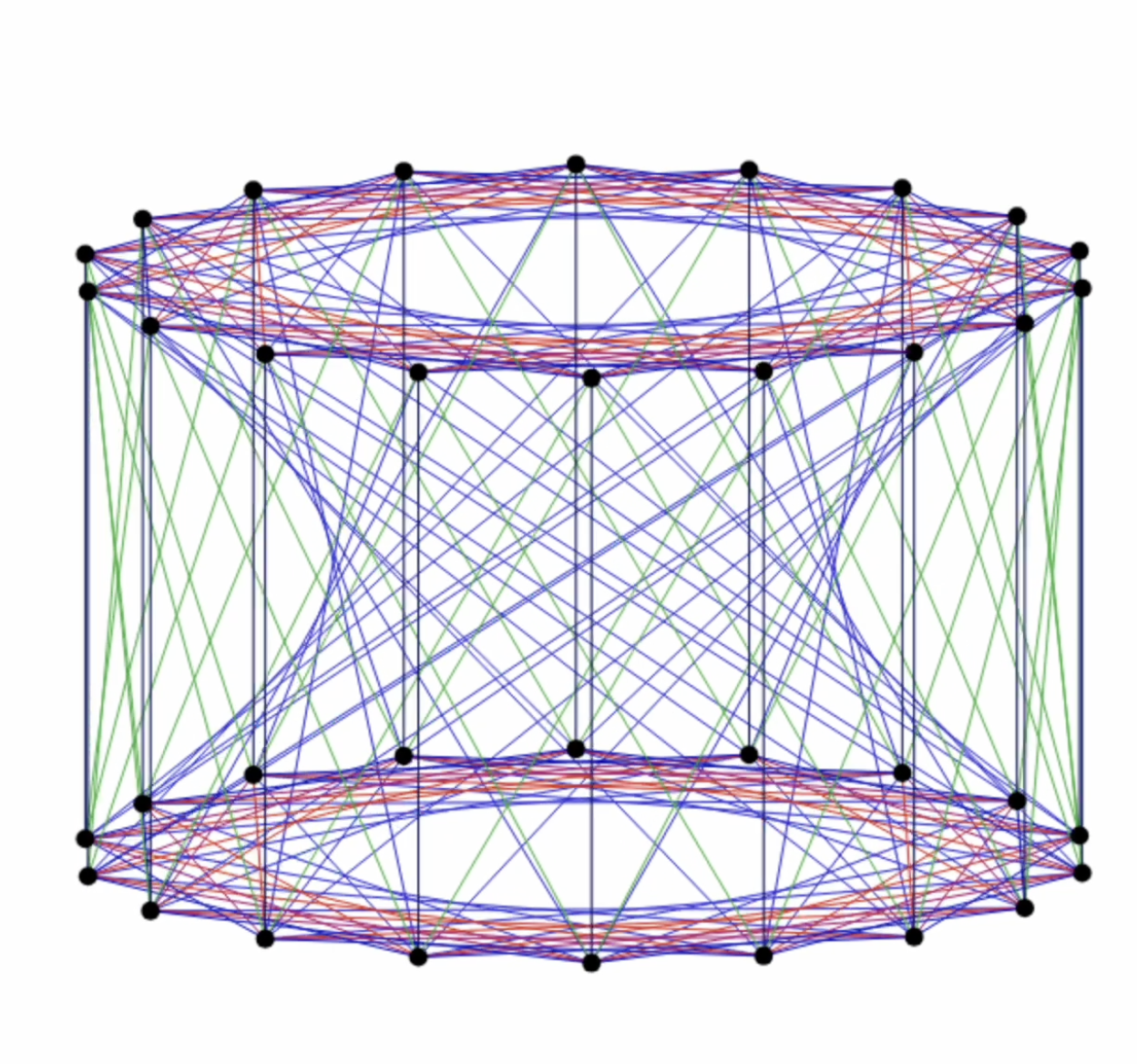}%
		}%
		\hfill
		\subcaptionbox*{ii.}[.45\linewidth]{
			\includegraphics[width=\linewidth]{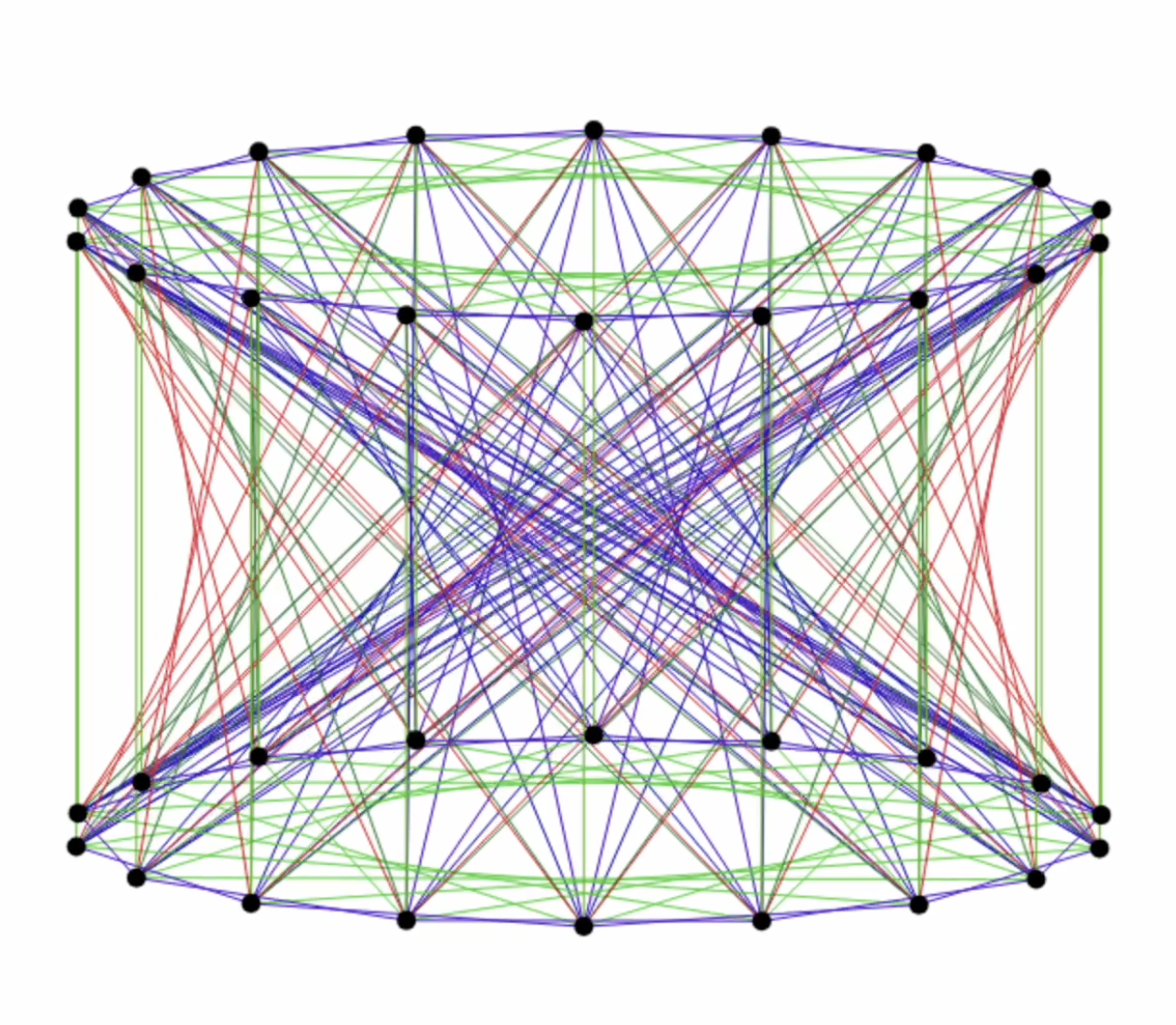}%
		}
		\caption{ i). MDC graph $\Gamma_{36, 1}$,  ii). MDC graph $\Gamma_{36, 2}$.}
	\end{figure}

	There is no circulant or cyclic construction for the optimal qubit code with parameters $\dsb{72, 0, 18}$. In Grassl's table~\cite{Grassl:codetables}, the optimal qubit code is listed by a stored generator matrix. We were able to obtain $16$ non-equivalent $\dsb{72, 0, 18}$ codes with the MDC construction. 
	
	\begin{proposition}
		There exists at least $16$ non-equivalent MDC based optimal qubit codes with parameters $\dsb{72, 0, 18}$. 
		The corresponding non-isomorphic MDC graphs have parameters
		$\Gamma_{72, i} =\Gamma({\bf N}, S_i ), \; 1\le i \le 16$, where 
		\begin{equation*}
			{\scriptsize
				\begin{split}
					S_1 &= [  \{ 2, 3, 5, 7, 8, 9, 10, 11, 12, 13, 15, 17, 18 \},
					\{ 1, 4, 8, 10, 11, 12, 13, 15, 18 \},
					\{ 1, 3, 17 \},	\{ 1, 2, 5, 7, 8, 9, 10, 12, 16 \}],\\
					S_2 &= [  \{ 3, 4, 6, 8, 12, 14, 16, 17 \},
					\{ 4, 5, 7, 13, 14, 16, 17 \},
					\{ 1, 6, 14 \},
					\{ 3, 4, 6, 7, 13, 15, 16 \}],\\
					S_3 &= [\{ 2, 4, 5, 6, 7, 13, 14, 15, 16, 18 \},
					\{ 1, 3, 4, 6, 7, 9, 10, 14, 15, 16, 17, 18 \},
					\{ 3, 5, 9, 10, 11, 15, 17 \},
					\{ 1, 2, 3, 4, 5, 6, 10, 11, 13, 14, 16, 17 \}]\\
					S_4  &=  [ \{ 5, 7, 8, 12, 13, 15 \},
					\{ 1, 2, 4, 5, 7, 8, 9, 12, 16, 18 \},
					\{ 2, 4, 5, 6, 8, 9, 10, 11, 12, 14, 15, 16, 18 \},
					\{ 1, 2, 4, 8, 11, 12, 13, 15, 16, 18 \}]	\\
					S_5 &= [  \{ 4, 5, 6, 8, 9, 11, 12, 14, 15, 16 \},
					\{ 3, 5, 6, 7, 8, 11, 12, 13, 14, 15, 16, 17 \},
					\{ 2, 3, 6, 9, 10, 11, 14, 17, 18 \},
					\{ 3, 4, 5, 6, 7, 8, 9, 12, 13, 14, 15, 17 \}]\\
					S_6 &= [\{ 2, 4, 6, 7, 9, 10, 11, 13, 14, 16, 18 \},
					\{ 2, 4, 7, 9, 11, 12, 16, 17 \},
					\{ 2, 5, 6, 8, 12, 14, 15, 18 \},
					\{ 3, 4, 8, 9, 11, 13, 16, 18 \}	]\\
					S_7 & = [\{ 5, 6, 7, 9, 11, 13, 14, 15 \},
					\{ 3, 6, 9, 11, 13, 16, 18 \},
					\{ 1, 2, 3, 4, 6, 7, 10, 13, 14, 16, 17, 18 \},
					\{ 2, 4, 7, 9, 11, 14, 17\}]\\
					S_8 & = [ \{ 2, 5, 6, 14, 15, 18 \},
					\{ 2, 3, 6, 11, 15, 16, 18 \},
					\{ 1, 6, 8, 12, 14 \},
					\{ 2, 4, 5, 9, 14, 17, 18 \} ]\\
					S_9 & =[ \{ 3, 4, 6, 7, 8, 9, 10, 11, 12, 13, 14, 16, 17 \},
					\{ 1, 4, 6, 7, 8, 10, 11, 12, 13, 14, 15, 17 \},
					\{ 7, 9, 11, 13 \},
					\{ 1, 3, 5, 6, 7, 8, 9, 10, 12, 13, 14, 16 \}]\\
					S_{10} & = [\{ 3, 4, 7, 8, 9, 11, 12, 13, 16, 17 \},
					\{ 2, 3, 4, 5, 8, 9, 13, 14, 15, 16, 17, 18 \},
					\{ 1, 4, 6, 7, 8, 9, 11, 12, 13, 14, 16 \},
					\{ 2, 3, 4, 5, 6, 7, 11, 12, 15, 16, 17, 18 \}]\\
					S_{11} & = [\{ 4, 8, 12, 16 \},
					\{ 1, 2, 3, 4, 6, 9, 10, 14, 17 \},
					\{ 2, 6, 7, 8, 9, 10, 11, 12, 13, 14, 18 \},
					\{ 1, 3, 6, 10, 11, 14, 16, 17, 18 \}]\\
					S_{12} & = [\{ 3, 4, 5, 6, 7, 10, 13, 14, 15, 16, 17 \},
					\{ 1, 2, 3, 4, 5, 12, 13, 16, 17 \},
					\{ 3, 4, 5, 15, 16, 17 \},
					\{ 1, 3, 4, 7, 8, 15, 16, 17, 18 \}]\\
					S_{13} &= [ \{ 2, 4, 6, 7, 8, 9, 11, 12, 13, 14, 16, 18 \},
					\{ 1, 2, 4, 6, 7, 8, 11, 16, 18 \},
					\{ 1, 2, 3, 4, 7, 8, 9, 11, 12, 13, 16, 17, 18\},
					\{ 1, 2, 4, 9, 12, 13, 14, 16, 18 \} ]\\
					S_{14} &= [\{ 2, 3, 4, 5, 7, 8, 9, 11, 12, 13, 15, 16, 17, 18 \},
					\{ 2, 3, 4, 5, 8, 10, 11, 12, 13, 14, 15, 16, 17, 18 \},
					\{ 2, 3, 4, 8, 9, 10, 11, 12, 16, 17, 18 \},
				 	\{ 2, 3, 4, 5, 6, 7, 8, 9, 10, 12, 15, 16, 17, 18 \}]\\
					S_{15} & = [ \{ 5, 6, 7, 9, 11, 13, 14, 15 \},
					\{ 5, 6, 7, 8, 9, 11, 12, 18 \},
					\{ 3, 10, 17 \},
					\{ 2, 8, 9, 11, 12, 13, 14, 15 \}]\\
					S_{16} & = [\{ 5, 6, 7, 9, 11, 13, 14, 15 \},
					\{ 1, 2, 4, 5, 6, 7, 9, 10, 11, 13, 14, 16, 18 \},
					\{ 3, 5, 7, 10, 13, 15, 17 \},
					\{ 1, 2, 4, 6, 7, 9, 10, 11, 13, 14, 15, 16, 18 \}] 
					\end{split} }
		\end{equation*}
		
	\end{proposition}
	
	\begin{proof}
		The graphs $\Gamma_{72, i}$ have valencies $[ 34, 25, 41, 39, 43, 35, 34, 25, 41, 45, 33, 35, 43, 53, 27, 41 ]$ for $1 \le i  \le16$ respectively. Therefore, graphs $\Gamma_1, \Gamma_2, \Gamma_3, \Gamma_4, \Gamma_5, \Gamma_6, \Gamma_{10}, \Gamma_{11}, \Gamma_{14}$, and $\Gamma_{15}$ are non-isomorphic. Let $C_{72, i}$ denote the corresponding additive self-dual codes from $\Gamma_{72, i}$ \; $1\le i \le 16$. Then $W_{18}(C_{72, i}) = [5760, 9108, 10404, 9768, 9244, 9684, 5028,8940, 9252, 9792, 9012, 8844, 7776,\\ 9336, 8256, 8064]$ implying codes $C_{72, i}$ are inequivalent for $1 \le i \le 16$.
	\end{proof}

	Further, there is no known circulant construction of the $\dsb{76, 0, 18}$ code. We were able to find $3$ inequivalent codes from the MDC construction.
	
	\begin{proposition}
		There are at least $3$ MDC based optimal non-equivalent qubit codes with parameters $\dsb{76, 0, 18}$. The corresponding MDC graphs are given by $\Gamma_{76, i} = \Gamma((2, 38), S_i), \; 1\le i \le 3$, where 
		\begin{equation*}
			{\scriptsize
				\begin{split}
					S_1   =  &[\{ 6, 8, 9, 10, 11, 29, 30, 31, 32, 34 \},
					\{ 3, 6, 9, 11, 13, 14, 16, 17, 20, 23, 24, 26, 27, 29, 31, 34, 37 \}]\\
					S_2  =  &[\{ 2, 3, 4, 5, 7, 8, 10, 11, 12, 13, 14, 15, 18, 22, 25, 26, 27, 28, 29, 30, 
					32, 33, 35, 36, 37, 38 \} ,\\
					&\{ 1, 2, 3, 4, 5, 6, 7, 8, 9, 11, 12, 13, 14, 15, 16, 17, 18, 22, 23, 24, 25,
					26, 27, 28, 29, 31, 32, 33, 34, 35, 36, 37, 38 \}]\\
					S_3  = &[\{ 2, 3, 5, 7, 8, 11, 12, 13, 15, 16, 17, 18, 19, 21, 22, 23, 24, 25, 27, 28,
					29, 32, 33, 35, 37, 38 \},\\
					&	\{ 1, 2, 4, 5, 6, 8, 11, 12, 13, 14, 15, 16, 17, 18, 22, 23, 24, 25, 26, 27, 
					28, 29, 32, 34, 35, 36, 38 \}]
				\end{split}
			}
		\end{equation*}
	\end{proposition}
	
	\begin{proof}
		The valence of the MDC graphs $\Gamma_{76,i}$ are $27, 29$, and $53$ respectively for $1\le i \le 3$. Hence, the graphs are non-isomorphic and the generated additive codes are non-equivalent. 
	\end{proof}

	Next, we introduce the two new qubit codes with lengths $77$ and $90$ obtained from the MDC construction. 
	
	\begin{proposition}
		The MDC graph $\Gamma_{77}$:=\mdc, where ${\bf N} = (7, 11)$ with the defining set $S = [s_1, s_2, \ldots, s_7]$, where 
		\begin{centering}
			$ s_1 = \{ 3, 4, 5, 6, 7, 8, 9, 10 \}$,
			$s_2 = \{ 1, 2, 4, 6, 10, 11 \}$,
			$s_3=\{ 1, 2, 3, 4, 5, 6, 9, 11 \}$,
			$s_4=\{ 1, 2, 4, 5, 6, 7, 9, 10, 11 \}$,
			$s_5=\{ 1, 2, 3, 4, 6, 7, 8, 9, 11 \}$,
			$s_6=\{ 1, 2, 4, 7, 8, 9, 10, 11 \}$, and 
			$s_7=\{ 1, 2, 3, 7, 9, 11 \}$
		\end{centering}
		generates a new $(77, 2^{77}, 19)_4$ additive self-dual code $C_{77}$. The corresponding qubit code $Q_{77}$ has parameters $\dsb{77, 0, 19} $ and exceeds the minimum distance of the best known qubit code listed on the Grassl's table~\cite{Grassl:codetables} by $1$. 
	\end{proposition}
	
	\begin{proof}
		The MDC graph $\Gamma_{77}$ and consequently, the additive code $C_{77}$ was found by a randomized search and the minimum distance of $C_{77}$ was verified using  {\tt MAGMA}.
	\end{proof}
	
	\begin{proposition}
		
		The MDC graph $\Gamma_{90} :=$ \mdc, where ${\bf N} = (9, 10)$ with the defining set $S = [s_1, s_2, \ldots, s_9]$ of vectors, where 
		\begin{centering}
			$s_1 = \{ 2, 3, 4, 5, 6, 7, 8, 9, 10 \}$,
			$s_2= \{ 1, 3, 4, 7, 9, 10 \}$,
			$s_3 = \{ 1, 2, 4, 6, 7, 8, 9 \}$,
			$s_4 = \{ 1, 3, 6, 8, 9, 10 \}$,
			$s_5 = \{ 1, 3, 5, 8, 9, 10 \}$,
			$s_6= \{ 1, 2, 3, 4, 7, 9 \}$,\\
			$s_7 = \{ 1, 2, 3, 4, 6, 9 \}$,
			$s_8=\{ 1, 3, 4, 5, 6, 8, 10 \}$, and
			$s_9 = \{ 1, 2, 3, 5, 8, 9 \}$
		\end{centering}
		produces a new $(90, 2^{90}, 22)_4$ Type $II$ additive self-dual code $C_{90}$. The new $\dsb{90, 0, 22}$ qubit code $Q_{90}$ has better minimum distance than the best known qubit code with parameters 
		$\dsb{90, 0, 21}$ given in~\cite{seneviratne}. 
		
	\end{proposition}
	
	\begin{proof}
		A randomized search among all symbol sets ${\bf N} = (2, 45), (3, 30), (5, 18), (6, 15)$ and $(9, 10)$ yielded the desired graph $\Gamma_{90}$ with the additive code $C_{90}$.  {\tt MAGMA}
		was used to verify the minimum distance of $C_{90}$. The code is Type $II$, since $|S| = 59$.
	\end{proof}
	

	\section{Concluding remarks}
	We have shown that MDC graphs have similar properties to circulant graphs and are as effective as circulant graphs with regard to obtaining optimal self-dual additive codes. Computationally, we observed that two-dimensional MDC graphs produced codes with higher minimum distances than 
		that of higher dimensional MDC graphs. This drawback may be due to the fact that  adjacency matrices of MDC graphs have  nested block circulant structures. 
	
\section*{Supplementary material}	

\begin{enumerate}
	\item MDC graphs $\Gamma_{77}$ and $\Gamma_{90}$, and their corresponding additive codes can be explicitly constructed by running {\tt MDB77Test.m} and {\tt MDB90Test.m}.
	\item The certificates of minimum distance computations for $\dsb{77, 0, 19}$ and $\dsb{90, 0, 22}$ are labeled {\tt New77output.txt }  and {\tt New90output.txt}.
\end{enumerate}

	\bibliographystyle{plain}
	
\end{document}